%% file: main.tex
\documentclass[aps,pra,10pt,preprintnumbers,twocolumn,amsmath,amssymb,nofootinbib,superscriptaddress]{revtex4-2}

\usepackage{graphicx}
\usepackage{dcolumn}
\usepackage{bm}
\usepackage{bbold}
\usepackage{calc}
\usepackage{qcircuit}
\usepackage{changes}
\usepackage{xcolor}
\usepackage{braket}
\usepackage{amsthm}
\usepackage[colorlinks=true,citecolor=cyan]{hyperref}
\hypersetup{colorlinks=true,citecolor=cyan,linkcolor=red,urlcolor=magenta}
\usepackage{algpseudocode}

\newtheorem{theorem}{Theorem}
\newtheorem{corollary}{Corollary}
\newtheorem{lemma}{Lemma}
\newtheorem{remark}{Remark}
\newtheorem{proposition}{Proposition}

\usepackage{color}
\usepackage{xcolor,colortbl}
\usepackage{comment}
\newcommand {\nc} {\newcommand}
\nc {\OA} [1]{\textcolor{black}{#1}}
\nc {\EJR} [1]{\textcolor{blue}{#1}}
\nc {\IB} [1]{\textcolor{blue}{#1}}
\nc {\IM} [1]{\textcolor{magenta}{#1}}
\nc {\LKC} [1] {\textcolor{orange}{[LEX: #1]}}

\newtheorem{definition}{Definition}

\usepackage[capitalise]{cleveref}
\crefname{section}{Sec.}{Secs.}
\crefname{appendix}{App.}{Apps.}
\crefname{table}{Tab.}{Tabs.}
\crefname{figure}{Fig.}{Figs.}
\crefname{definition}{Def.}{Defs.}
\crefname{lemma}{Lem.}{Lems.}
\crefname{theorem}{Thm.}{Thms.}
\crefname{corollary}{Cor.}{Cors.}
\crefname{remark}{Rem.}{Rems.}
\crefname{proposition}{Prop.}{Props.}

\begin{document}

\title{CaRBM: A Fixed-Depth Quantum Algorithm with Partial Correction for Thermal State Preparation}


\author{Omar Alsheikh}
\email{ooalshei@ncsu.edu}
\affiliation{Department of Physics and Astronomy, North Carolina State University, Raleigh, North Carolina 27695, USA}

\author{A.~F.~Kemper}
\email{akemper@ncsu.edu}
\affiliation{Department of Physics and Astronomy, North Carolina State University, Raleigh, North Carolina 27695, USA}

\author{Ermal Rrapaj}
\affiliation{Lawrence Berkeley National Laboratory, One Cyclotron RD, Berkeley, CA, 94720, USA}
\affiliation{Department of Physics, University of California, Berkeley, CA 94720, USA}
\affiliation{RIKEN iTHEMS, Wako, Saitama 351-0198, Japan}

\author{Evan J. Rule}
\email{erule@lanl.gov}
\affiliation{Theoretical Division, Los Alamos National Laboratory, Los Alamos, NM 87545, USA}

\author{Goksu C. Toga}
\email{gctoga@ncsu.edu}
\affiliation{Department of Physics and Astronomy, North Carolina State University, Raleigh, North Carolina 27695, USA}

\preprint{LA-UR-26-22202}

\begin{abstract}
We introduce the CaRBM algorithm for fixed-depth thermal state preparation. Our algorithm is based on thermal state purification and uses the Restricted Boltzmann Machine (RBM) block-encoding scheme to implement the imaginary-time propagator $e^{-\beta H}$, which is implemented in the quantum circuit in a fixed-depth manner via Cartan decomposition. Our algorithm performs best at high temperatures, with the success probability of the block encoding decreasing as the temperature decreases. To increase the success probability, we have devised a correction scheme for the block-encoding that increases the temperature range our algorithm reliably probes. We demonstrate our algorithm by calculating the partition function zeros of the XXZ model and the phase diagram of the Gross-Neveu model, which is a model of strongly interacting relativistic fermions.
\end{abstract}
\maketitle

\section{Introduction}

State preparation is a key step in almost any quantum simulation, often contributing a large portion of the cost of the simulation. The typical task for quantum algorithms is the preparation of a pure state.  However, in a realistic setting any quantum system will eventually couple to an environment, and to simulate such open systems one needs to be able to prepare mixed states. Of particular interest are thermal Gibbs states, which describe open quantum systems held at a constant temperature. This observation has led to the development of thermal state preparation algorithms that have been applied to relevant problems with varying degrees of success in quantum simulation \cite{childs2018toward}, quantum machine learning \cite{kieferova2017tomography,biamonte2017quantum}, and combinatorial optimization \cite{kirkpatrick1983optimization,somma2008quantum}.

\begin{figure}[t!]
    \centering
    \includegraphics[width=\columnwidth]{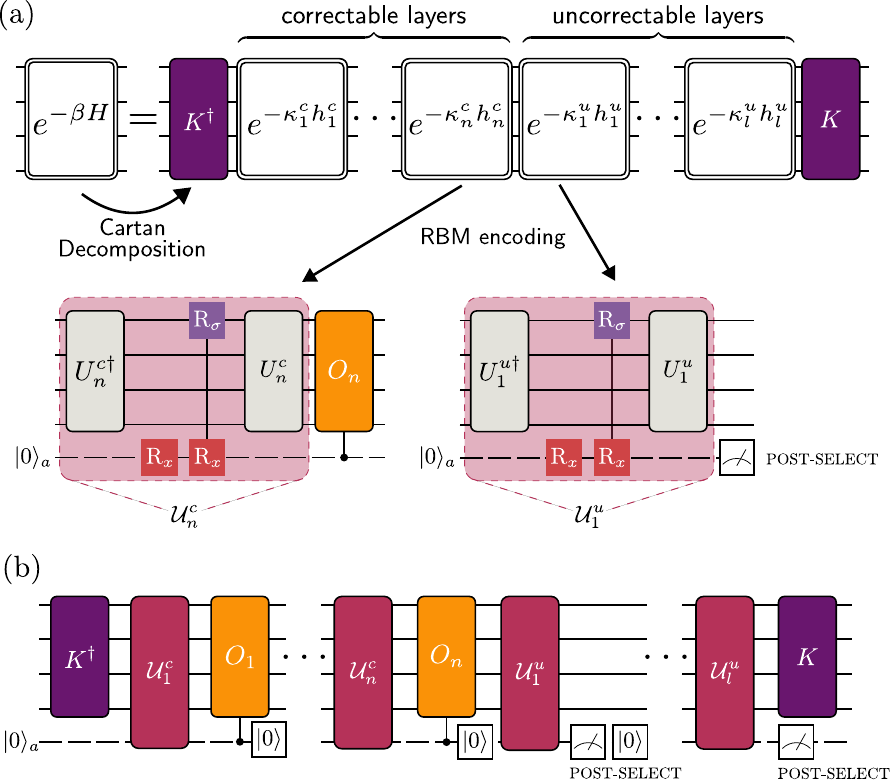}
    \caption{(a) The three building blocks of CaRBM are shown. First, the Cartan decomposition is used to transform the imaginary time propagator $e^{-\beta H}$ into a product of commuting propagators. These commuting propagators are then encoded via RBM into unitary operations by coupling the system to an ancilla. The first few layers can be corrected by adding a controlled operator on the ancilla, ensuring a perfect success rate, while the subsequent layers are encoded probabilistically by post-selecting on the ancilla. (b) The full circuit for $e^{-\beta H}$ using CaRBM.  
    }
    \label{fig:figure1}
\end{figure}

The existing techniques for thermal state preparation come with different limitations. Common issues with variational algorithms for state preparation are the expressibility of the ansatz, the high dimensional classical optimization necessary, and barren plateaus~\cite{larocca2025barren,wang2021noise}. Thermal state preparation also comes with additional challenges, mainly due to the fact that the cost function needed to generate the thermal state must minimize the free energy, requiring expensive entropy measurements. This introduces an additional hurdle on top of the already existing ones~\cite{chowdhury2020variational,sagastizabal2021variational,getelina2023adaptive}.
Variational algorithms that avoid the expensive entropy measurements have been recently proposed, but they all come with their limitations, such as being limited to low temperatures~\cite{sambasivam2025tepid}.

Alternative approaches that avoid these challenges are based on imaginary time evolution (ITE). These methods --- also known as purification --- are frequently employed in classical computing methods, like tensor networks, as well as in quantum algorithms~\cite{feiguin2005finite,feiguin2005time,verstraete2004matrix}. The main drawback of such quantum algorithms is that the operator needed for ITE is not unitary, which makes it unimplementable on a quantum device without further treatment like block encoding or approximating the effects of the non-unitary operations by unitaries \cite{love2020cooling,low2019hamiltonian,gilyen2019quantum,rrapaj2025exact,simon2025ladder,mcardle2019variational,motta2020determining}.

There are two main approaches to employing ITE on a quantum computer. The first is the Quantum Imaginary Time Evolution (QITE) algorithm~\cite{mcardle2019variational,motta2020determining,hejazi2024adiabatic,
gomes2020efficient}, which reproduces the non-unitary effect of the ITE operator by breaking down the evolution into small steps; \OA{at each timestep, a unitary that approximates the normalized state after an ITE step can be found by solving a least squares problem.}
\OA{However, the computational cost quickly grows with each step as the correlation length of the state increases due to the need for more complicated (non-local) unitaries to approximate the ITE process}.

The second approach is to embed the non-unitary operator inside a larger unitary matrix, then recover the effects of the non-unitary operation by post-selection. This is known as block-encoding \cite{low2019hamiltonian,gilyen2019quantum,leadbeater2024non,rrapaj2025exact,simon2025ladder}, and it is our method of choice in this paper.  
We employ a block-encoding scheme --- based on a neural network architecture known as restricted Boltzmann machine (RBM) --- that implements the non-unitary action on the physical (visible) qubits as part of a larger unitary operator by expanding the system with additional ancillary (hidden) qubits~\cite{rrapaj2025exact} (see \cref{fig:figure1}). In the original algorithm, operators of the form $e^{-\kappa \sigma}$, where $\sigma$ is a Pauli string, are expressed as a sum of unitary operators, which can be realized as a quantum circuit with an ancillary qubit. The operation is then accepted with some acceptance probability depending on the outcome of the measurement on the ancilla.

Similar to standard unitary synthesis, the goal now becomes finding a way to compile a quantum circuit for $e^{-\beta H}$, where $H$ is a given Pauli sentence. Common techniques rely on series expansions~\cite{Berry2015} or product formulas~\cite{lloyd1996universal,haah2021quantum,Childs2021theory}, but they require increasing circuit depth for low temperatures (i.e., long evolution times with large $\beta = 1/T$), which is especially problematic in the Noisy Intermediate Scale Quantum (NISQ) and early fault tolerant eras. Furthermore, in addition to decoherence issues and lack of practical error-correction protocols, having more exponents $e^{-\kappa \sigma}$ to implement equates to having to do more post-selection, exponentially diminishing the total acceptance probability of the process.

To circumvent these issues, we introduce the CaRBM algorithm for thermal state preparation. We implement the ITE operator via the Cartan decomposition, which uses the Lie algebraic properties of the Hamiltonian to compile circuits with depths that are independent of simulation time~\cite{kokcu2022fixed,alsheikh2025redcard}. The Hamiltonian gets ``rotated" into an Abelian algebra (i.e., a set of commuting Pauli strings) that leads to an exact product formula of exponents of Pauli strings that can then be readily implemented via the RBM encoding. Using a fixed-depth architecture for post-selection techniques has the added benefit of fixing the total number of post-selection processes one needs to perform, ensuring a better total success rate.

We improve upon previous implementations of the RBM encoding by demonstrating that the first few layers --- on the order of the number of qubits ---  of ITE can be modified to ensure a perfect success probability, at the expense of one controlled Pauli-string operation per layer (see \cref{fig:figure1}).  This further motivates the use of a fixed-depth architecture, ensuring maximum gain from the correction process.

The result is a quantum circuit that generates any finite-temperature state with a fixed-depth cost with respect to the temperature. Our algorithm uses an infinite-temperature Thermo Field Double (TFD) state followed by ITE to generate a thermal state at any temperature, requiring no minimization procedure and hence no entropy calculations~\cite{feiguin2005finite,feiguin2005time}. CaRBM further improves upon existing non-variational imaginary-time purification techniques that are based on Suzuki-Trotter decomposition by providing a fixed-depth implementation of the ITE operator that also employs partial error correction.

The layout of the paper is as follows: In Sec.~\ref{sec:methods}, we outline the Cartan decomposition, the RBM encoding
and the correction scheme. In Sec.~\ref{sec:examples}, we demonstrate our algorithm by (1) calculating the partition function zeros of the $\rm XXZ$ model and (2) probing the finite density and temperature phase diagram of a model of strongly interacting relativistic fermions. We conclude with a summary of our work in Sec.~\ref{sec:summary}.

\vspace{-2mm}
\section{Methods}\label{sec:methods}
In this section, we outline the building blocks of the CaRBM algorithm.
\OA{First, we discuss how we use the Cartan decomposition to obtain a fixed-depth product formula. We then explain how we encode the propagator using the RBM architecture, which allows us to perform the ITE as a sequence of unitary operations.}
Finally, we explain our correction scheme, which drastically improves the sampling rate, allowing us to probe temperature regions that were not accessible with the uncorrected algorithm.

\vspace{-4mm}
\subsection{Cartan decomposition}

We want to express the imaginary propagator $e^{-\beta H}$, where $H=\sum_i H_i \sigma^i$, and $H_i$ are real coefficients and $\sigma^i$ are Pauli strings, as a product of exponentials of single Pauli strings. \OA{Every exponent in the product is going to be separately encoded into a unitary that performs the desired non-unitary operation with some success probability. Naturally, the more exponents one has to encode, the smaller the total success rate becomes. Because of this, widely used approximation protocols like the Suzuki-Trotter approximation can be problematic, especially for long simulation times (or equivalently, low temperatures). This motivates us to employ a protocol that results in circuits whose depth does not increase as simulation time grows. Our method to express the propagator in the desired form is the Cartan decomposition, which uses the Lie algebraic structure of the Hamiltonian to express it as}

 \begin{align}
     H = K h K^\dagger,
 \end{align}
 where $h$ is an element in an Abelian algebra called a Cartan subalgebra, i.e., it is a linear combination of commuting Pauli strings, and
 \begin{align}
     K = \prod_j e^{i\theta_jk^j},
 \end{align}
 where $k^j$ are Pauli strings and $\theta_j$ are real parameters \OA{which can be determined by finding a local extremum of an effective cost function \cite{earp2005constructive,kokcu2022fixed,alsheikh2025redcard}.} 

 The propagator becomes
 \begin{align}
     e^{-\beta H} = K e^{-\beta h} K^\dagger.
 \end{align}
 Since the terms in $h$ commute, we can break up the exponent into a product of exponentials of single Pauli strings without incurring any Trotter error. We then proceed to implement the ITE via RBM for each term individually.

\vspace{-2mm}
\subsection{Finite-temperature block encoding}
\label{sec:block_enc}

Let $\sigma_{(J)}$ be a Pauli string of length $M$ that acts on the subset of qubits $J=\{j_1,\ldots,j_M$\}. The exponential of any such string can be converted to an exponential of single-qubit Pauli operators,
\begin{equation}
    \exp\left[- \kappa \sigma_{(J)}\right]=U \exp\left[-\kappa Z_{(J)}\right]  U^{\dagger},
\end{equation}
where $Z_{(J)}\equiv\sum_{j\in J}Z_j$, and the unitary $U$ is a product of single-qubit gates that transform from the $Z$ basis to the $X$ or $Y$ basis (e.g., Hadamards) and two-qubit entangling gates (e.g., CX's). In short, a Pauli string of arbitrary complexity is equivalent (modulo Hadamards, CX's, etc.) to a single-qubit Pauli string in the corresponding computational basis. This result greatly simplifies our block-encoding scheme; rather than constructing a unitary encoding of the imaginary-time propagator of arbitrary Pauli strings, we only need to find a unitary encoding of the imaginary-time propagator of single-qubit Pauli operators. 

Introducing an ancillary qubit $a$ initialized in the $\ket{0}$ state, the imaginary-time propagator acting on physical qubit $j$  in initial state $\ket{\psi}$ can be expressed as \cite{rrapaj2025exact}
\begin{align}
e^{-\kappa \sigma_j}\ket{\psi} &= A\sum_{\tilde{h}=\pm 1} e^{-i(W\sigma_j + bI)\tilde{h}}|\psi\rangle \nonumber \\
&= 2A\bra{0}_a e^{-i(W \sigma_j+ b I )\otimes X_a} \ket{\psi}\otimes\ket{0}_a,
\label{eq:2body_id}
\end{align}
where $\sigma_j\in \{I_j,X_j,Y_j,Z_j\}$. The interaction between the physical qubit and the ancilla, including the marginalization over the latter, represents a neural network architecture known as a restricted Boltzmann machine (RBM). The RBM parameters $W$ and $b$ are real-valued, and therefore the operator on the total system (physical plus ancillary) is unitary. Post-selection is successful if the ancillary qubit remains in the initial state upon measurement. (An equivalent version of Eq. \eqref{eq:2body_id} is obtained if the ancillary qubit is initialized and post-selected in state $\ket{1}$.) There exist infinitely many encoding functions $W(\kappa)$, $b(\kappa)$ that produce the same imaginary-time propagator up to an overall normalization $A(\kappa)$. The specific parameterizations that we employ in our numerical calculations are provided in \cref{app:RBMparameters}. 
 
 When acting on a generic state, one drawback of block encoding is its probabilistic nature: In the event of a post-selection failure --- when the action of the RBM operator flips the ancilla bit --- the entire circuit must be restarted, and the sample collected is not useful. Different choices of encoding functions $W(\kappa)$, $b(\kappa)$ lead to different post-selection success probabilities depending on the initial state $\ket{\psi}$ of the physical qubit (see App. \ref{app:RBMparameters} for details). We can, however, correct post-selection for the first several layers of imaginary-time propagators in our circuit, resorting to probabilistic encoding for subsequent layers.

\vspace{-4mm}
 \subsection{Correction scheme for the improving sampling rate}
 In general, the post-selection success rate decays exponentially in the coupling $\pm\kappa$, with the sign depending on the choice of RBM parameters and the initial state $\ket{\psi}$. Since the failed results are discarded, an exponential number of resources are required in order to accurately implement the ITE. As mentioned earlier, it is possible to use some of the failed results to guarantee a perfect success rate regardless of the measurement on the ancilla. This is made possible by a specific choice of the parameters $W, b$ that implements $e^{\kappa \sigma}|\psi\rangle$ instead of $e^{-\kappa\sigma}|\psi\rangle$ when the post-selection fails (see \cref{app:RBMparameters}).

Let us see how this works. Pick a unitary operator $O$ such that $\{O,\sigma\}=0$ and $O|\psi\rangle \sim |\psi\rangle$. Then $Oe^{\kappa\sigma}|\psi\rangle \sim e^{-\kappa\sigma}|\psi\rangle$. To implement this, we can apply a controlled $O$ gate on the post-selection ancilla qubit. If the readout is 0, then ITE succeeded and we do not need to correct for anything. If the readout is 1, then $O$ is applied and we get the required ITE. For one ITE operation, this simple procedure ensures a 100\% success probability by adding one controlled operation to the circuit. Considering the fact that the success probability decreases exponentially with $|\kappa|$, this is a significant improvement on quantum resources. For implementation purposes, we restrict $O$ to a Clifford operation, i.e., we choose $O$ to be a Pauli string.

If we want to correct more than one layer, we can follow a similar logic. If we have 4 layers, for example, and we want to fix the leftmost layer, we would want to pick $O$ such that it switches the sign of this layer while preserving the sign of the first three. We state this observation as the following theorem, and refer the reader to \cref{app:corrections} for its proof and more details.
\begin{theorem}
     For the ITE process $e^{-\kappa_j\sigma^j}e^{-\kappa_{j-1}\sigma^{j-1}}\dots e^{-\kappa_1\sigma^1}|\psi\rangle$, we can find an operator $O$ that anticommutes with $\sigma^j$ but commutes with all the other layers, provided that $\sigma^j$ cannot be written as a product of the other Pauli strings. The $j$th layer can thus be corrected provided that the operator $O$ that fixes the layer satisfies $O|\psi\rangle \sim |\psi\rangle$.
\end{theorem}
The theorem relies on two conditions. The first one, that is, the commutation properties of $O$ with the $\sigma$'s, puts an upper bound on the total number of layers that can be corrected. We show in \cref{app:corrections} that this number cannot exceed the number of physical qubits. Nevertheless, if it exists, we can systematically find such $O$ with ease. The second condition, that is $O|\psi\rangle \sim |\psi\rangle$, is in general more difficult to satisfy for a generic $|\psi\rangle$. However, since we are starting with the fully mixed state to prepare thermal states, this condition is going to be trivially satisfied, since the fully mixed state is proportional to the identity matrix, and thus it commutes with any operator.

\vspace{-4mm}
\section{Applications}\label{sec:examples}
\vspace{-3mm}
In this section, we use our algorithm to calculate thermal observables that are of interest to condensed matter and high-energy physics; specifically, locating partition function zeros and mapping finite-temperature and -density phase diagrams of relativistic fermions. Each of these cases involves a variation of the CaRBM algorithm, which are illustrated in Fig.~\ref{fig:circuits}.  While they all start with an infinite temperature TFD state generated by the circuit in panel (a), the subsequent pieces are slightly different. Panel (b) depicts the circuit used to find the zeros of the partition function in the complex plane when the system is probed with a complex external field (these are called Lee-Yang zeros). By coupling the system with an additional ancilla and time evolving the coupled system, the partition function can be found by measuring the ancilla. If we complexify the temperature instead of adding a complex probe field, we get Fisher zeros through the circuit shown in panel (c) Note that in this case the circuit simplifies and the entirety of $K$ can be avoided. Finally, we can simply prepare a thermal state and measure the physical qubits to find thermal averages as shown in panel (d).

\begin{figure}[t]
    \centering
    \includegraphics[clip=true, trim=0 0 0 0, width=0.72\columnwidth]{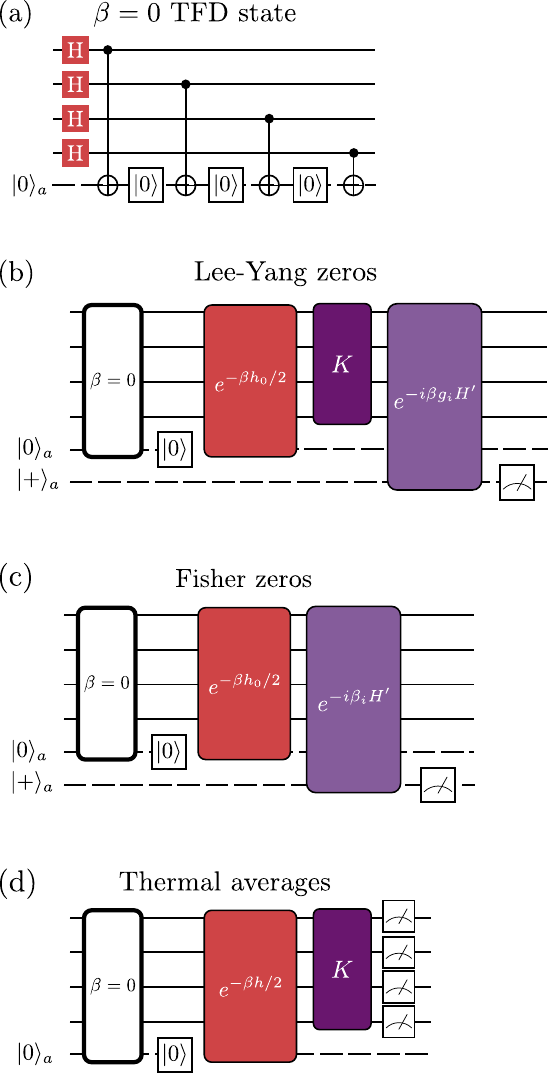}
    \caption{Schematics for the different circuits used for (a) infinite-temperature thermal state preparation, (b) Lee-Yang zeros, (c) Fisher zeros, and (d) thermal averages.}
    \label{fig:circuits}
\end{figure}

\vspace{-4mm}
\subsection{Partition function zeros}
\vspace{-2mm}
We start by demonstrating that our algorithm can be used to calculate partition function zeros, which have become one of the cornerstone observables in thermal systems due to recent developments allowing their measurement in experimental systems and the close relationship they exhibit to dynamical and regular phase transitions~\cite{yang1952statistical,wei2012lee,wei2014phase,liu2024exact,liu2023signatures,francis2021many}.
There are two different versions of the partition function zeros that we will be interested in: The first is the construction of Lee and Yang, where instead of complexifying the temperature, we couple our system to a complexified probe field. The second is the original definition due to Fisher, where the temperature is allowed to be complex.

\begin{figure*}[t]
    \centering
    \includegraphics[width=0.95\textwidth]{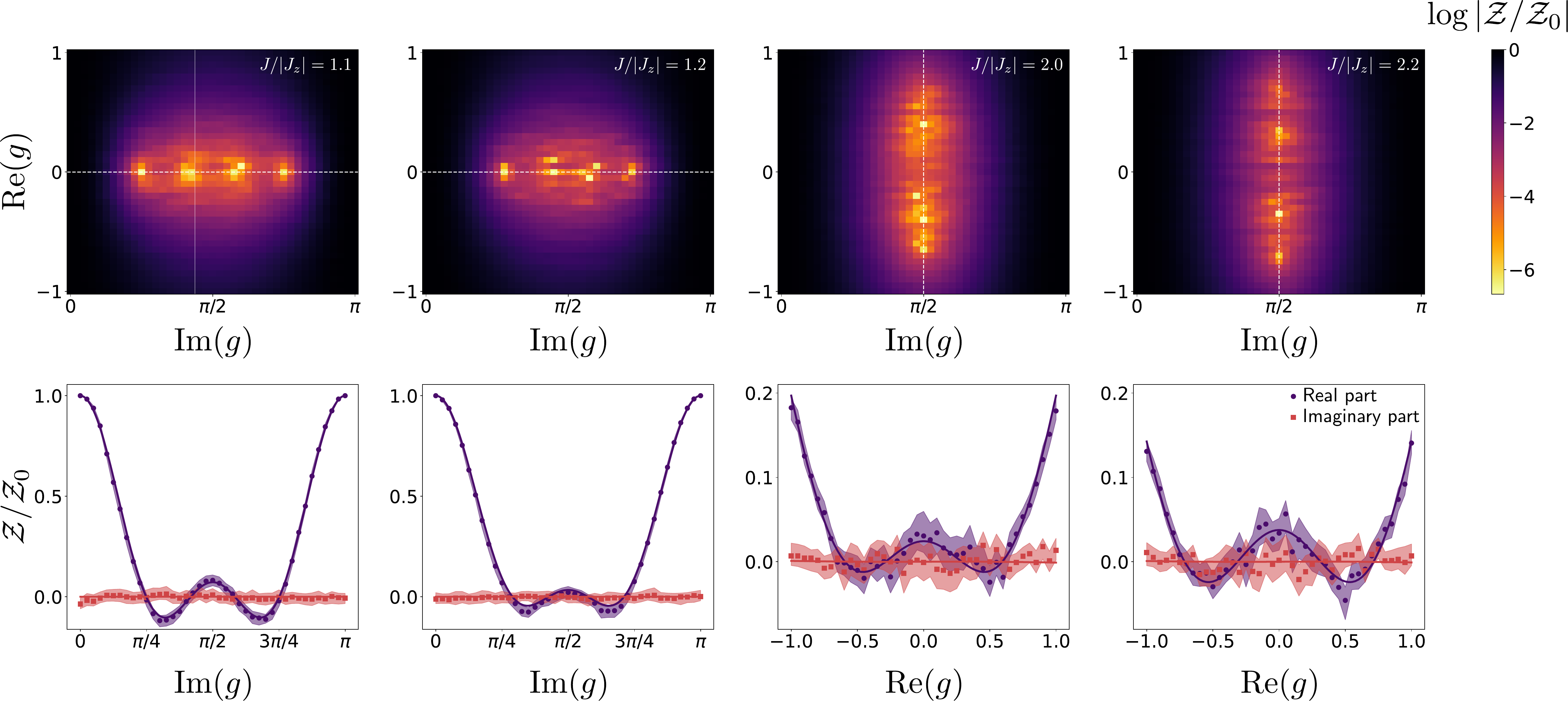}
    \caption{Top: Circuit simulations of partition function plots showing the locations of Lee-Yang Zeros for the 4-site XXZ model coupled to a complexified probe field. The behavior of the zeros undergoes a considerable change as we increase the value of the ratio $J/|J_z|$, signifying a phase transition from the Ising-like phase to the XY-like phase. Bottom: Line plots of sections taken on the heat maps. The first two are taken at $\mathrm{Re}(g) = 0$ and the second two are taken at $\mathrm{Im}(g) = \pi/2$. The sections are chosen to pass through the zeros. The markers are the circuit simulation results, and the shaded regions are confidence intervals at a confidence level of 0.95, estimated via bootstrap. The solid lines are MPS calculations. All simulations are at $\beta = 1$.}
    \label{fig:LYzeros}
\end{figure*}

In their seminal work, Lee and Yang investigated the partition function of an Ising model in the presence of a complexified magnetic field and showed that the locations of the zeros change with respect to temperature, and at critical temperatures, they pinch the real axis showing that the location of zeros could be used to identify criticality in quantum systems \cite{yang1952statistical}. Lee-Yang zeros have not been treated as a real observable since they only occur at imaginary magnetic field or temperature until Wei and Liu proposed a method to experimentally measure Lee-Yang zeros using the decoherence of a probe spin coupled to the Ising system \cite{wei2012lee}. Here, we will use the same approach of coupling the system to a probe spin and measuring its coherence to calculate the partition function zeros.

The first step in calculating the Lee-Yang zeros is to introduce a complexified probe field, $H_B = g\sum_{i}^{L}Z_i = gH_I$, where $g=g_r+ig_i$. Given the system Hamiltonian $H_s$, the partition function is
\begin{align}
    \mathcal{Z}(\beta,g) = \mathrm{Tr} \exp{\left(-\beta H_0 -i\beta g_iH_I\right)},
\end{align}
where $H_0 = H_s + \mathrm{Re}\ H_B$. In the case when $H_s$ and $H_I$ commute, we can compute the partition function by coupling the system to an ancilla with the coupling Hamiltonian
\begin{align}\label{eq:int_H}
    H' = \frac{1}{2} H_I \otimes \sigma_i^z.
\end{align}
We then prepare the initial state
\begin{align}\label{eq:initial_state}
    \rho(0) = \frac{e^{-\beta H_0}}{\mathcal{Z}_0(\beta,g_r)} \otimes |+\rangle \langle +|,
\end{align}
where $\mathcal{Z}_0(\beta,g_r)$ is the partition function of $H_0$. The system is then time-evolved under $U(\beta g_i) = \exp \left(-i\beta g_i H' \right)$. The physical system's qubits are traced out, and a measurement of the off-diagonal element of the resulting density matrix yields $\mathcal{Z}(\beta,g)/ \mathcal{Z}_0(\beta,g_r)$. For a general $H_I$, a more complex $H'$ is needed \cite{wei2014phase}. 

We use CaRBM to prepare the state given by \cref{eq:initial_state}. We perform the KHK decomposition on $H_0$ to arrive at $H_0 = K_0h_0K_0^\dagger$.
We rewrite \cref{eq:initial_state} as
\begin{align}
    \rho(0) = K_0 \frac{e^{-\beta h_0}}{\mathcal{Z}_0(\beta,g_r)} K_0^\dagger\otimes |+\rangle \langle +|.
\end{align}
We obtain the thermal state at inverse temperature $\beta$ for $h_0$ via imaginary time evolution starting from the infinite-temperature thermal state, i.e., the normalized identity matrix \OA{A circuit construction for the fully mixed state is shown in \cref{fig:circuits}(a)}.

After the desired thermal state at a certain inverse temperature $\beta$ is obtained, we act on the state by the unitary $K_0$. We then evolve this state in real-time under $U(\beta g_i)$ \OA{\cref{fig:circuits}(b) illustrates the circuit used to obtain the Lee-Yang zeros}.

Following the procedure described above we calculate the Lee-Yang zeros of the $\rm XXZ$ model at $\beta=1$. The system Hamiltonian of size $L$ is given by
\begin{align}
    H_s = -J\sum_{i=1}^{L-1} \left(X_i X_{i+1} + Y_i Y_{i+1}\right) - J_z\sum_{i=1}^{L-1} Z_i Z_{i+1}.
\end{align}

\begin{figure*}[t]
    \centering
    \includegraphics[width=0.95\textwidth]{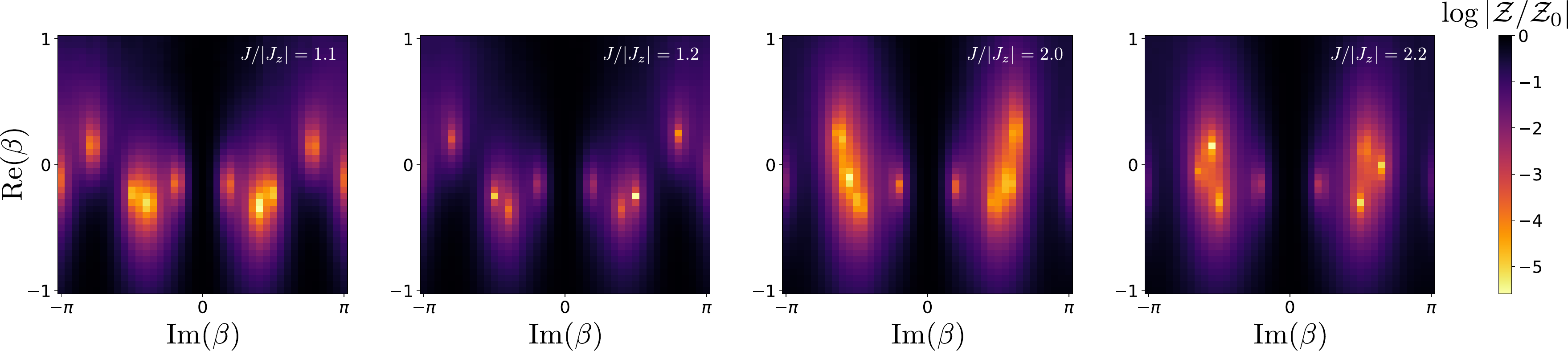}
    \caption{Complex partition function showing the locations of Fisher Zeros for 4-site XXZ model obtained via simulations of the circuits discussed in the main text.}
    \label{fig:Fzeros}
\end{figure*}

We work within the ferromagnetic Ising regime $J_z > 0$, and we probe the partition function zeros between an $\rm XY$-like regime ($|J_z| \ll |J|$) and an Ising-like regime ($|J_z| \gg |J|$). We use the complex probe field
\begin{align}
    H_B = g\sum_{i=1}^L Z_i,
\end{align}
which commutes with $H_s$. A choice for the Cartan subalgebra is
\begin{widetext}
\begin{align}\label{eq:CSA}
    \mathfrak{h} = \mathrm{span}_{i\mathbb{R}}\{Z_{j_1}Z_{j_2}\dots Z_{j_r} \ | \ 1\leq r \leq L, \ 1\leq j_1 < \dots < j_r \leq L \ \backslash \ Z_1\dots Z_L \},
\end{align}
\end{widetext}
\OA{where $\mathrm{span}_{i\mathbb{R}}$ denotes the span of $i$ times the Pauli strings over the real numbers. \cref{eq:CSA} contains all products of $Z$'s except $Z_1\dots Z_L$.} Our Hamiltonian $H_0$ is transformed into $h_0$ such that $ih_0 \in \mathfrak{h}$.

Our results in \cref{fig:LYzeros} show the expected distinct behavior of the Lee-Yang zeros in the two phases of the $\rm XXZ$ model: In the Ising-like regime, the Lee-Yang zeros lie on the horizontal $g_r=0$ line. After crossing the phase transition to the XY-like regime, the zeros switch to being on a vertical $g_i=\pi/2$ line \cite{francis2021many}.

\OA{In the bottom row, we plot the real and imaginary parts of the partition function at $g_r=0$ for the Ising regime and $g_i=\pi/2$ for the XY regime. The plots show a confidence level of 0.95, estimated via bootstrap. Note that while the XY regime plots may look noisier, it is a result of the different scale. In fact, the Ising regime plots are noisier on average.}

To probe the Fisher zeros, we complexify the inverse temperature $\beta = \beta_r + i\beta_i$. This is equivalent to taking $H_I = H_0$, which means that \cref{eq:int_H} always applies. We evolve the thermal state under the unitary $U(\beta_i) = \exp \left(-i\beta_i H'\right)$. Due to $H_I=H_0$, the $K$ unitary completely drops due to the cyclic property of the trace, leaving us with a shallower circuit, depicted in \OA{\cref{fig:circuits}(c)}. In \cref{fig:Fzeros} we again plot the logarithm of the absolute value of the partition function. We see that Fisher zeros also show some dependence on the ratio $J/|J_z|$. However, it is more difficult to spot a qualitative change as we move across the two phases as was in the Lee Yang case. \OA{Nevertheless, we can still use the Fisher zeros to reconstruct the partition function. The implementation of Fisher zeros is also sometimes simpler, particularly because a Cartan decomposition for $H_0=H_s$ can be cheaper than a Cartan decomposition for the system with a probe field $H_0=H_s+\mathrm{Re}\ H_B$. It can also result in a smaller Cartan subalgebra, giving higher success rates.}

\subsection{Thermal phase diagram of relativistic fermions in $1+1d$ }
In general, any finite-density simulation of an interacting fermionic system will exhibit a sign problem, which makes the use of Monte Carlo sampling techniques very hard to implement. A natural solution for this is to use quantum computers where the sign problem is not an issue. However, there is still the question of the thermal state preparation which needs to be addressed.  As a second example of the CaRBM algorithm, we will reproduce the phase diagram of a strongly interacting relativistic fermionic model 
called the Gross-Neveu model \cite{gross1974dynamical} at finite temperature and density with open boundary conditions.

The Gross-Neveu model, due to its asymptotic freedom and spontaneous chiral symmetry breaking, has attracted a great deal of attention in recent years as a toy model for simulating relativistic interacting fermions on quantum computers \cite{roose2021lattice,asaduzzaman2022quantum,bakalov2026quantum,akiyama2023matrix,kong2026phase}.
 The continuum Hamiltonian of the Gross-Neveu model is the massive Dirac equation 
 \begin{equation}
    H_0=\int dx\; i\bar{\psi} \alpha\partial_x \psi+m \bar{\psi}\beta\psi \label{Dirac_H},
\end{equation}
plus a $4$-fermion interaction term. 
\begin{equation}
    H_G = \int dx\; \frac{-G^2}{2}(\bar{\psi}\psi)^2.
\end{equation}
 For discretization, we use the staggered fermion formalism \cite{kogut1975hamiltonian,susskind1977lattice} and fix  $\alpha = X$, $\beta = Z$. 
 After staggerization and Jordan-Wigner transformation, the corresponding Hamiltonian in terms of spin operators is given as, 
\begin{align} \label{qubit_ham}
H^{(N)}_{0,m}=&\sum_{a=1}^{N} \frac{1}{2}\bigg[ i\sum_{n=1}^{L-1} \Big(S_n^+(a)S_{n+1}^-(a)-S_n^-(a)S_{n+1}^+(a)\Big)\nonumber\\
&+m\sum_{n=1}^{L}\left(-1\right)^n\left(S_n^-(a)S_n^+(a)\right) \bigg]+{\rm h.c} \nonumber\\
=& \sum_{a=1}^{N} \Bigg[ \sum_{n=1}^{L-1} \Big( -X_n(a)Y_{n+1}(a)+Y_n(a)X_{n+1}(a)\Big) \nonumber\\
&+m \sum_{n=1}^{L}\left(-1\right)^n\left(1-Z_n(a)\right) \Bigg],
\end{align}
where $S^\pm = \frac{1}{2}\left(X \pm iY\right)$, and the four-fermion interaction term is
\begin{equation}
    H^{(N)}_{G}=-\frac{1}{2}G^2\sum_{n=1}^{L} \,\, \sum_{a=1}^{N} \,\, \sum_{b,b>a} (I-Z_n(a))(I-Z_n(b)).
\end{equation}
Here, $N$ denotes the number of flavors and $L$ the number of lattice sites. To probe the phase diagram at finite density, we add a chemical potential term 
\begin{equation}
    H_{\mu} =  \int dx\; \mu \bar{\psi}\gamma\psi .
\end{equation}
Since we already fixed $\alpha$ and $\beta$ only choice remaining for the $\gamma$, which anticommutes with $\alpha,\beta$ and squares to the identity, is $\gamma=Y$.  This leads to the  following chemical potential term after the Jordan-Wigner transformation takes the following form 
\begin{equation}
    H_{\mu}^{(N)} =\mu \sum_{a=1}^{N} \sum_{n=1}^{L-1}i(-1)^n (S_n^+(a)S_{n+1}^-(a)-S_n^-(a)S_{n+1}^+(a))
\end{equation}

We then set $m=0$ to simplify the Gross-Neveu Hamiltonian and focus only on the chemical potential term. Our Hamiltonian of interest is given as, 
\begin{equation}
    H_{GN}^{(N)} = H_{0}^{(N)} + H_G^{(N)} + H_{\mu}^{(N)}
\end{equation}
Further details regarding the discretization of the model and details about the chemical potential term can be found in \cref{sec:GN}.

In our simulations, we set $N=2$ and $L=4$
and calculate the two-point correlation function $\sum_i\langle \bar{\psi}_i\psi_0\rangle \sim \sum_i \langle Z_iZ_0\rangle$. Again, we start with the infinite temperature state and obtain the desired state via imaginary time evolution using $e^{-\beta H_{GN}}$. 
We find that the phase diagram obtained with our method which can be seen in Fig.~\ref{fig:phase_diag_GN} is very similar to the known results, even for small lattice sizes and number of flavors~\cite{wolff1985phase,lenz2020inhomogeneous,lenz2022inhomogeneities}. As expected, we see a symmetry-broken phase with a finite condensate for large $\beta$ and small $\mu$ indicated as the bright region in the graph, and a symmetric phase for large $\mu$ indicated as the black region.  However, we don't see a sign of an inhomogeneous phase due to the small lattice volumes and number of flavors. The effect of the correction scheme is also apparent in these phase diagram plots. In the top left plot, we clearly see an unphysical phase on the top right. After implementing the correction scheme, we see that this region disappears, and we obtain the expected phase diagram. 

The effect of the correction scheme becomes more apparent when we look into the success rate of the post-selection process these plots can be seen in the bottom row of Fig.~\ref{fig:phase_diag_GN}, where the black color which represents a very low probability of success covers a much bigger portion of the phase space before the correction and is much more confined to top right corner of the phase space.

\begin{figure}[t]
    \centering
    \includegraphics[width=\columnwidth]{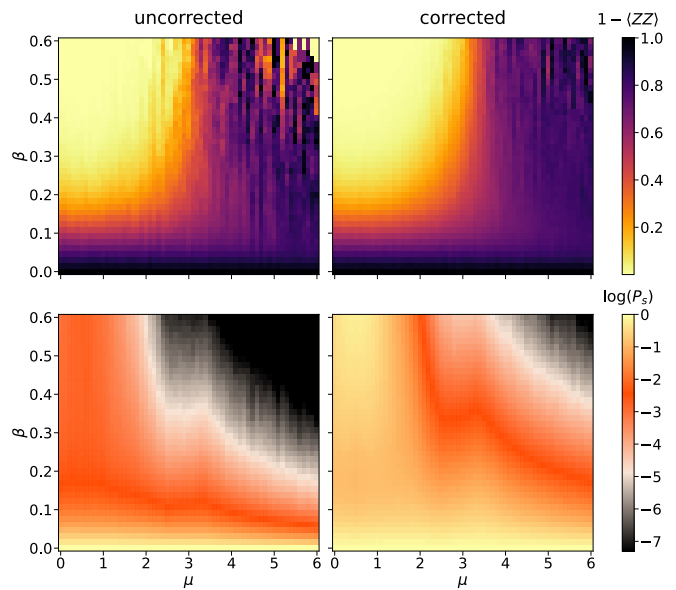}
    \caption{Top: Circuit simulations $\sum_i \langle Z_iZ_0\rangle$ for the 2-site Gross-Neveu model with 2 fermion flavors across a range of temperatures $\beta$ and chemical potentials $\mu$. Simulations without and with correction layers are shown for comparison. Bottom: Success rate of the post-selection process. The correction procedure allows us to explore an extended region in the phase space that was otherwise inaccessible (the blackened region) due to vanishing success probabilities.
    }
    \label{fig:phase_diag_GN}
\end{figure}

This shows that our algorithm can be successfully implemented to obtain the phase diagrams of relativistic fermions where Monte Carlo approaches fail due to the sign problems. 

\section{Conclusion and Outlook}
\label{sec:summary}
In this paper, we introduced the CaRBM quantum algorithm for fixed-depth thermal state preparation. We use the Cartan decomposition to bring the Hamiltonian into an Abelian basis, which we use to find a product formula for  the ITE propagator $e^{-\beta H}$. The individual propagators obtained are then fed into the RBM block encoding machinery to probabilistically encode the ITE as a sequence of unitaries. To prepare a thermal state at an inverse temperature $\beta$, we start with the fully mixed state then evolve it under $e^{-\beta H/2}$ using CaRBM.
The combination of these two methods allows us to partially correct some layers, considerably improving the total success rate of the post-selection process. This allows us access to lower temperature regimes.

To demonstrate our algorithm, we ran quantum circuit simulations to find the partition function zeros of the $\rm{XXZ}$ model, which are in excellent agreement with the tensor network results.
 As a second example, we calculated the phase diagram of the Gross-Neveu model, which is a model of strongly interacting relativistic fermions at finite temperature and density. We showed that our algorithm reproduces the expected phase diagram. Finite density calculations of strongly interacting systems are of interest to many high-energy and nuclear physics problems, but the classical approaches to simulating these systems usually lead to a sign problem, hindering any Monte-Carlo approach to simulate such systems. We showed that our algorithm can capture the phase diagram of the Gross-Neveu model successfully and can be used for further studies of such models which is of high interest to high-energy and nuclear physics communities.

\OA{The Cartan decomposition greatly complements the RBM encoding in another aspect as well. We observe that the coefficients of the Pauli strings that appear in the Abelian element $h$ are almost always skewed, i.e., a large portion of the norm of $h$ is concentrated in a handful of coefficients, even if the strings that appear in the original $H$ have the same coefficients. Since we can only correct the first few layers, choosing the layers that have the largest absolute coefficients to correct becomes more impactful, because the following post-selection on the remaining ITE processes have small coefficients in the exponent.}

\OA{One drawback of this type of algorithm is the large overhead of finding the Cartan decomposition. The problem entails a classical optimization problem with dimension that increases exponentially with system size \cite{alsheikh2025redcard}, unless the system is fast-forwardable \cite{khaneja2005constructive,earp2005constructive,drury2008quantum_shannon,daugli2008general}. The algorithm also suffers from diminishing acceptance ratios as $\beta$ grows, which is a general feature of block-encoding algorithms. Other techniques (e.g. Ref.~\onlinecite{sambasivam2025tepid}) can then be used for lower temperatures in conjunction with CaRBM, which performs very well at high temperature regimes.}

\OA{Moving forward, we note that the post-selection could be improved by choosing a different RBM block encoding, or by optimizing the correction procedure. 
The correction procedures may also be adapted to work in circuit structures whose depth varies with the problem.
} 

\vspace{0.2in}

\section*{Acknowledgments}
AFK acknowledges financial support from the U.S. Department of Energy, Office of Science, Advanced Scientific Computing
Research under DE-SC0025623 and GCT and OA from the National Science Foundation under award No. PHY-2325080: PIF: Software-Tailored Architecture for Quantum Co-Design. ER is supported by the U.S. Department of Energy (DOE) under Contract No. DE-AC02-05CH11231, through the National Energy Research Scientific Computing Center (NERSC), an Office of Science User Facility located at Lawrence Berkeley National Laboratory. EJR was supported by the National Science Foundation under cooperative agreement 2020275 during partial completion of this work. EJR is supported by the U.S. Department of Energy through the Los Alamos National Laboratory. The research presented in this article was partially supported by the Laboratory Directed Research and Development program of Los Alamos National Laboratory under project numbers 20251163PRD3 and 20260043DR. Los Alamos National Laboratory is operated by Triad National Security, LLC, for the National Nuclear Security Administration of U.S. Department of Energy (Contract No. 89233218CNA000001).

\input{main.bbl}
\newpage
\appendix

\section{Block-encoding parameters}
\label{app:RBMparameters}
As discussed in Sec. \ref{sec:block_enc}, the unitary encoding of the imaginary-time propagator introduced in Eq. \eqref{eq:2body_id} is specified by parameters $(W,b)$ and normalization constant $A$. These parameters depend on the value of the imaginary-time parameter $\kappa$, but not on the Pauli matrix being encoded, $\sigma_j=X_j,Y_j,Z_j$. One valid choice of parameters is
\begin{equation}
    \begin{split}
        A&=\frac{1}{2}e^{|\kappa|},\\
        W&=\frac{1}{2}\cos^{-1}\left(e^{-2|\kappa|}\right),\\
        b&=sW,
    \end{split}
\end{equation}
where $s=\mathrm{sign}(\kappa)$. Properly implementing the imaginary-time propagator via Eq. \eqref{eq:2body_id} requires the ancillary qubit to be measured in the initial state after application of the block-encoded unitary. The probability of success depends on the initial state of qubit $l$, the value of the imaginary-time parameter $\kappa$, and the specific choice of encoding parameters $(W,b)$. To calculate it, we can expand the unitary evolution that appears on the right-hand side of \cref{eq:2body_id} as
\begin{align}
    e^{-i(W\sigma_j+bI)\otimes X_a} |\psi\rangle |0\rangle_a = &\cos{(W\sigma_j + b I)} |\psi \rangle \otimes |0\rangle \nonumber \\ 
    &- \mathrm{i}\sin{(W\sigma_j + b I)} |\psi \rangle \otimes |1\rangle.
\end{align}
The ITE is encoded in the first term of the sum, and the success probability is given by
\begin{align}
    P_s = \langle\cos^2{(W\sigma_j + b I)}\rangle = 1-\langle\sin^2(W\sigma_j+bI)\rangle.
\end{align}
For the above parameter choices, the probability of success is
\begin{equation}
    P_s=1-\left(1-e^{-4|\kappa|}\right)\alpha,
\end{equation}
where $\alpha$ is the probability of measuring $|\psi\rangle$ in the $\sigma_j=s$ state. In particular, if the initial state is such that $\alpha=0$, then the block-encoding is guaranteed to succeed. 

Alternatively, we can use the parameterization
\begin{align}\label{eq:correctionparam}
    A&=\sqrt{\frac{\cosh(2\kappa)}{2}}, \nonumber\\
    W&=\tan ^{-1}\left(e^{2 \kappa}\right) - \frac{\pi}{4} \nonumber\\ 
    b &=s\frac{\pi}{4},
\end{align}
with a probability of success
\begin{align}
    P_s = 1 - \frac{1}{1+e^{4s\kappa}} - \tanh(2s\kappa) \alpha,
\end{align}
where $\alpha$ is the probability of measuring $|\psi\rangle$ in the $\sigma_j=1$ state. On average, this choice of parametrization yields a lower success rate, particularly for $|\kappa| < 1$. Still, this choice is useful because when the post-selection fails, we get the (non-normalized) evolution
\begin{align}
    \sin\left(W\sigma_j + \frac{\pi}{4}I\right) = \cos\left(-W\sigma_j + \frac{\pi}{4}I\right).
\end{align}
Noticing that $W(-\kappa) = -W(\kappa)$ for this parametrization, we conclude that the failed trial results in the application of $e^{\kappa\sigma_j}$ instead of the desired $e^{-\kappa\sigma_j}$. In the next section we lay out how we can correct this failure to have a perfect success rate.

\section{Layer Corrections}
\label{app:corrections}
If we use the parametrization \cref{eq:correctionparam}, the success probability is usually smaller. However, when post-selection fails, instead of getting $e^{-\kappa\sigma}$ we get $e^{\kappa\sigma}$. Normally, the failed results are discarded, but in this particular instance we can still use them. Let us say that we want to find the evolution $e^{-\kappa\sigma} |\psi\rangle$. If we fail, we instead get $e^{\kappa\sigma}|\psi\rangle$. Pick a unitary operator $O$ such that $\{O,\sigma\}=0$ and $O|\psi\rangle \sim |\psi\rangle$. Then $Oe^{\kappa\sigma}|\psi\rangle \sim e^{-\kappa\sigma}|\psi\rangle$. To implement this, we can apply a controlled $O$ gate on the post-selection ancilla qubit. If the readout is 0, then ITE succeeded and we do not need to correct for anything. If the readout is 1, then $O$ is applied and we get the required ITE. For one ITE operation, this simple procedure ensures a 100\% success probability by adding one controlled operation to the circuit. Considering the fact that the success probability decreases exponentially with $\beta$, this is a significant improvement on quantum resources. For implementation purposes, we restrict $O$ to a Clifford operation, i.e., we choose $O$ to be a Pauli string. In many cases, we can correct more than one layer. In this section, we provide a formal treatment for the correction process and derive a theoretical limit on the number of layers that can be corrected for a given Hamiltonian.

We start by introducing a lemma regarding the multilayered correction process:
\begin{lemma}\label{lem:multlayercorrection}
    Suppose that we have a sequence of ITE operators acting on some state $e^{-\kappa_j\sigma_j}\dots e^{-\kappa_1 \sigma_1} |\psi\rangle$. If one can find a Pauli string $O$ such that $\{O,\sigma_j\}=0$ and $[O,\sigma_r]=0$ for $r < j$, as well as $O|\psi\rangle \sim |\psi\rangle$, then the $j$th layer can be corrected by added a control $O$ gate.
\end{lemma}
\noindent\textbf{Proof} Suppose the post-selection fails at the $j$th exponent. Then one would have applied $e^{\kappa_j\sigma_j}$ instead of the desired $e^{-\kappa_j\sigma_j}$. One can then apply $O$ to correct the error:
\begin{multline}
    Oe^{\kappa_j\sigma_j}e^{-\kappa_{j-1}\sigma_{j-1}}\dots e^{-\kappa_1\sigma_1}|\psi\rangle \\
    \sim e^{-\kappa_j\sigma_j}e^{-\kappa_{j-1}\sigma_{j-1}}\dots e^{-\kappa_1\sigma_1}|\psi\rangle.
\end{multline}

Since $O$ commutes with the first $j-1$ terms, it can pass through without changing anything. Since we only need to apply $O$ when post-selection fails, i.e., when we measure the ancilla in the $|1\rangle$ state, we can add an $O$ gate to the circuit, controlled with the ancilla qubit.\hfill $\blacksquare$

The existence of such operator $O$ is not guaranteed by \cref{lem:multlayercorrection}. In order to prove (or disprove) its existence, we need to do more work. 

To proceed, we introduce the binary symplectic representation of Pauli strings:

\begin{definition}[Binary Symplectic Representation]
    A Pauli string $\sigma$ of length $n$ can be encoded as 2 bit vectors $\vec{a}$ and $\vec{b}$, each of length $n$. Each Pauli matrix $\sigma^i$ at position $i$ is encoded as follows:
    \begin{center}
    \begin{tabular}{c|c|c}
        $\sigma^i$ & $a^i$ & $b^i$ \\
        \hline
        $I$ & 0  & 0 \\
        $X$ & 1  & 0 \\
        $Z$ & 0 & 1 \\
        $Y$ & 1 & 1
   \end{tabular}
   \end{center}
\end{definition}
Notice that the unsigned product of two Pauli matrices is simply given by addition modulo 2. Any Pauli string can thus be represented as a bit vector $\vec{p} = (\vec{b}^T| \vec{a}^T)^T$. For example, $XIZZY$ is encoded as $(00111 |10001)^T$. The power of this representation is the ease of finding unsigned products and commutation relations of Pauli strings. We state how the unsigned product is found.
\begin{remark}\label{rem:product}
    Given two Pauli strings $\sigma_1$ and $\sigma_2$ with the binary symplectic forms $\vec{p}_1=(\vec{b}_1^T|\vec{a}_1^T)^T$ and $\vec{p}_2 = (\vec{b}_2^T|\vec{a}_2^T)^T$, respectively, the unsigned product is given by
    \begin{align}
        \vec{p} \equiv \vec{p}_1 + \vec{p}_2 \mod 2.
    \end{align}
\end{remark}
The representation of the unsigned product as a vector sum is also convenient when we have a set of Pauli strings such that one of them can be written as a product of the others.
\begin{remark}\label{rem:associative}
    Given a set of distinct Pauli strings $\sigma_1, \sigma_2, \dots \sigma_r$ with binary symplectic forms $\vec{p}_1,\dots,\vec{p}_r$ such that $\sigma_r = \sigma_1\dots \sigma_{r-1}$ (up to a phase), then $\dim (\mathrm{span}_{\mathbb{Z}_2}\{\vec{p}_1,\dots,\vec{p}_r\}) = r - 1$, i.e., the vectors are linearly dependent.
\end{remark}
\cref{rem:associative} will prove useful later on. We are also interested in the commutation relations between Pauli strings, so we provide a lemma that expresses the commutator between two Pauli strings in the binary symplectic form.
\begin{lemma}\label{lem:comrelations}
    Given two Pauli strings $\sigma_1$ and $\sigma_2$ with the binary symplectic forms $(\vec{b}_1^T|\vec{a}_1^T)^T$ and $(\vec{b}_2^T|\vec{a}_2^T)^T$, respectively, then one can check whether the commute by looking at the following equation:
    \begin{align}
        \vec{a}_1 . \vec{b}_2 + \vec{b}_1 . \vec{a}_2 \equiv
        \begin{cases}
            0 \mod 2\quad \text{if } [\sigma_1,\sigma_2] = 0, \\
            1 \mod 2\quad \text{if } \{\sigma_1,\sigma_2\} = 0.
        \end{cases}
    \end{align}
\end{lemma}

\noindent\textbf{Proof} First we look at the dot product. $a_1^i b_2^i$ is non-zero only if $a_1^i=b_2^i = 1$, which can only happen if $\sigma_1^i$ is $X$ or $Y$, and if $\sigma_2^i$ is $Z$ or $Y$. If both $\sigma_1^i$ and $\sigma_2^i$ are $Y$, then $b_1^i a_2^i=1$ as well, and we get a contribution of 2 to the sum, which will not affect the result modulo 2. For all the other cases, $b_1^ia_2^i$ will vanish, and we get a contribution of 1 to the sum. Similarly, if $\sigma_1^i$ is $Z$ or $Y$ and $\sigma_2^i$ is $X$ or $Y$, then $b_1^ia_2^i=1$. Unless both of them are $Y$, we get a contribution of 1 to the sum. The sum thus counts the number of different non-identity Paulis between $\sigma_1$ and $\sigma_2$. If the number is even, then $\sigma_1$ and $\sigma_2$ commute. Otherwise, they do not commute.\hfill $\blacksquare$

We have now completed all the necessary steps to look for a Pauli string with the required commutation relations. We provide the following theorem:
\begin{theorem}\label{thm:linearsystem}
    Suppose that we have a set of Pauli strings $\sigma_1,\sigma_2,\dots,\sigma_j$, then a Pauli string $O$ such that $\{O,\sigma_j\}=0$ and $[O,\sigma_r]=0$ for $r<j$ can be found by solving a system of linear equations.
\end{theorem}
\noindent\textbf{Proof} Let $\Sigma_a = (\vec{a}_1\ \dots\ \vec{a}_j)^T$, $\Sigma_b = (\vec{b}_1\ \dots\ \vec{b}_j)^T$, $\Sigma = (\Sigma_a\ \Sigma_b)$, and let $\vec{O} = (\vec{b}^T|\vec{a}^T)^T$ be the binary symplectic representation of $O$. By \cref{lem:comrelations}, the desired commutation relations take the form of a system of $j$ linear equations:
\begin{align}\label{eq:linearsystem}
    \Sigma \vec{O} \equiv \vec{c} \mod 2,
\end{align}
where $\vec{c} = (00\dots 0 1)^T$, and $\Sigma$ is a $j\times 2n$ matrix, where $n$ is the number of qubits we are working with. If any solution exists, then there will be $2^{2n - \mathrm{rank}(\Sigma)}$ different solutions.\hfill$\blacksquare$

To ensure that a solution exists, no inconsistencies between the linear equations can arise. Since $\vec{c}$ only has 1 non-zero component, only inconsistencies involving the last row, i.e., $\sigma_j$ can arise. Specifically, the only way we run into an inconsistency is if the last row is linearly dependent on any of the other rows, or, equivalently, if $\sigma_j$ is a product of a subset of the previous $j-1$ Pauli strings (see \cref{rem:associative}). We state this as the following corollary:
\begin{corollary}
    A solution $O$ exists if and only if $\sigma_j$ is not algebraically related to the other Pauli strings, i.e., it cannot be written as a product of any of the other $j-1$ Pauli strings.
\end{corollary}
We also state the main result:
\begin{corollary}\label{cor:condition}
    The $j$th layer of the ITE process $e^{-\kappa_j\sigma_j}e^{-\kappa_{j-1}\sigma_{j-1}}\dots e^{-\kappa_1\sigma_1}|\psi\rangle$ can be corrected, provided that $\sigma_j$ is not algebraically related to the other Pauli strings, and provided that any of the solutions satisfies $O|\psi\rangle \sim |\psi\rangle$.
\end{corollary}
\cref{thm:linearsystem} also provides a theoretical limit on the number of layers than can be corrected. If the ITE layers are sequentially corrected, then by \cref{cor:condition}, there can be no algebraic relations between the corrected Pauli strings. This in turn means that the coefficient matrix $\Sigma$ (see \cref{eq:linearsystem}) will always have full rank, i.e., at the $rth$ step, $\mathrm{rank}(\Sigma)=r$. This leads to the following:
\begin{corollary}
    The maximum number of layers that can be corrected in the ITE process $e^{-\kappa_j\sigma_j}e^{-\kappa_{j-1}\sigma_{j-1}}\dots e^{-\kappa_1\sigma_1}|\psi\rangle$ is $2n$, where $n$ is the number of qubits.
\end{corollary}

In our algorithm, we find a KHK decomposition for the Hamiltonian, which allows us to tranform the ITE generator into an element of an Abelian Lie algebra. This further lowers the maximum number of layers that could possibly be corrected. We state, without proof, a useful fact about the Cartan subalgebras of $\mathfrak{su}(2^n)$:

\begin{proposition}\label{prop:ZCSA}
    Any maximal Abelian subalgebra -- a Cartan subalgebra -- of $\mathfrak{su}(2^n)$ is equivalent, under a similarity transformation, to the subalgebra spanned by tensor products of $Z$ Pauli matrices.
\end{proposition}

This statement is useful because it tells us that we can have at most $n$ commuting Pauli strings that are not algebraically related. This leads to the final result of the section:
\begin{corollary}
    For commuting generators, the maximum number of layers that can be corrected in the ITE process $e^{-\kappa_j\sigma_j}e^{-\kappa_{j-1}\sigma_{j-1}}\dots e^{-\kappa_1\sigma_1}|\psi\rangle$ is $n$, where $n$ is the number of qubits.
\end{corollary}

\section{Gross-Neveu Model}
\label{sec:GN}
To get to the discretized GN model we start with the general form of the continuum Hamiltonian for one Dirac fermion:
\begin{equation}
    H=\int dx\; i\bar{\psi} \alpha\partial_x \psi+m \bar{\psi}\beta\psi 
\end{equation}
To discretize this we replace 
$\partial\to \Delta_{n^\prime\;n}=\frac{1}{2}\left(\delta_{n^\prime\;n+1}-\delta_{n^\prime\;n-1}\right)$ and spin diagonalize on lattice with the unitary
transformation $\psi(n)=(\alpha)^n\lambda(n)$ following Kogut-Susskind~\cite{kogut1975hamiltonian,susskind1977lattice} formulation where 
\begin{equation}
    H=\sum_n \,{\lambda}^\dagger(n)\frac{i}{2}\left[\lambda(n+1)+\lambda(n-1)\right]+
    m{\lambda}^\dagger(n)\beta\lambda(n)\end{equation}

Next, we allocate the 2 spin components to even and odd lattice sites corresponding to the choice $\alpha=X$, which gives us the kinetic term as 
\begin{equation}
    \lambda_1^\dagger(n)(\lambda_2(n+1)-\lambda_2(n-1)) + \lambda_2(n)^\dagger(\lambda_1(n+1)-\lambda_1(n-1))
\end{equation}

 For the mass term, we can choose any matrix for $\beta$ that anti-commutes with $\alpha$ and squares to one.
The conventional choice is $\beta=Z$, which generates the following mass term
\begin{equation}
    \lambda_1^\dagger(n)\lambda_1(n) - \lambda_2^\dagger(n)\lambda_2(n) 
\end{equation}

If we denote $\chi_{\rm even}=\lambda^1$ and $\chi_{\rm odd}=\lambda^2$ this can be trivially written  as $m\left(-1\right)^n \lambda^\dagger(n)\lambda(n)$ (i.e the mass is
plus or minus m corresponding to the eigenvalue of $Z$) and the whole action becomes 
\begin{align}
    H=\frac{1}{2}\sum_{n}\chi^\dagger(n)\left[\chi(n+1)+\chi(n-1)\right]\\
    +m\sum_n\left(-1\right)^n\chi^\dagger(n)\chi(n) \nonumber
\end{align}
For two flavors of Dirac fermion equipped with a four-fermion term  the Hamiltonian looks
like
\begin{align}
    H=\frac{1}{2}\sum_{n}\sum_{a=1}^2\chi^{a\dagger}(n)\left[\chi^a(n+1)+\chi^a(n-1)\right]\\ 
    -G^2\chi^{1\dagger}(n)\chi^1(n)\chi^{2\dagger}(n)\chi^2(n) \nonumber
\end{align}
which is identified as the $N_f=2$ Gross-Neveu model in one spatial dimension.~\cite{gross1974dynamical}
Now we add  a chemical potential that will mix $\lambda_1$ and $\lambda_2$  
\begin{equation}
    H=\int dx\; i\bar{\psi} \alpha\partial_x \psi+m \bar{\psi}\beta\psi +\mu\bar{\psi}\gamma\psi 
\end{equation}

Notice that the usual choice of the chemical potential which is $\bar{\psi}\gamma_0\psi$ would just map into the mass, and instead we choose $\gamma=Y$, another choice could be $\gamma=\gamma_5$.~\cite{susskind1977lattice} Starting with $\gamma=Y $ we get 
\begin{equation}
    \bar{\psi}\gamma_5\psi = \lambda_1^\dagger\lambda_2 -\lambda_2^\dagger \lambda_1 = (-1)^n(\chi^\dagger(n) \chi(n+1) - h.c.)
\end{equation}

Unlike the usual mass term, this term mixes $\lambda_1$ and $\lambda_2$ and reproduces the expected GN phase diagram. 

Another aspect of this term is that while it still breaks the chiral symmetry of the model in the continuum limit,  it respects the sublattice symmetry of the model which can be realized by even shifts of the lattice index, which is related to the application of $\gamma_5$ operator to our staggered fermions.~\cite{roose2021lattice}

To simulate this system on a quantum computer we first need to rewrite the theory
in terms of Pauli matrix or qubit operators. We
use the Jordan-Wigner transformation~\cite{jordan1928paulische,dargis1998fermionization}
\begin{align}
    \chi_n(a)&= \prod_{b<a}P^{\left(\sigma^b\right)}(L) \prod_a  P^{\left(\sigma^a\right)}(n-1) S_+^a(n),
\end{align}
where 
\begin{equation}
    P^{\left(\sigma^a\right)}(n)=\prod_{y=1}^nZ^a(y)
\end{equation}
and $S_\pm=\frac{1}{2}\left(X\pm iY\right)$ and $\sigma^a$ denotes the Pauli matrices. It is straightforward, 
to show that this representation respects the fundamental anticommutator required for
fermion operators
\begin{equation}
    [\chi^{a\dagger}(x),\chi^b(y)]_+=\delta_{xy}\delta^{ab}.
\end{equation}

\end{document}

%% file: main.bbl
%